\documentclass[10pt, conference]{IEEEtran}
\IEEEoverridecommandlockouts

\usepackage{cite}
\usepackage{amsmath,amssymb,amsfonts}
\usepackage{algorithmic}
\usepackage{amsthm}
\usepackage{graphicx}
\usepackage{textcomp}
\usepackage{xcolor}
\usepackage{subcaption} 

\newtheorem{definition}{Definition}
\newtheorem{proposition}{Proposition}
\newtheorem{thm}{Theorem}
\def\BibTeX{{\rm B\kern-.05em{\sc i\kern-.025em b}\kern-.08em
    T\kern-.1667em\lower.7ex\hbox{E}\kern-.125emX}}
\begin{document}

\title{Lightweight Multi-Vehicle Collaborative Perception Acceleration with Fusion Position Adjustment
\thanks{The work presented in this paper is funded by the National Natural Science Foundation of China  No. 62201079, the Beijing Natural Science Foundation No. L232051.}
}


\author{\IEEEauthorblockN{Wenzhao Zhang\IEEEauthorrefmark{1}\IEEEauthorrefmark{3}, Shujun Han\IEEEauthorrefmark{2}\IEEEauthorrefmark{3}, Haixiao Gao\IEEEauthorrefmark{1}\IEEEauthorrefmark{3}, Mengying Sun\IEEEauthorrefmark{1}\IEEEauthorrefmark{3}, Bizhu Wang\IEEEauthorrefmark{1}\IEEEauthorrefmark{3}, Xiaodong Xu\IEEEauthorrefmark{1}\IEEEauthorrefmark{3}\IEEEauthorrefmark{4}}
\IEEEauthorblockA{\IEEEauthorrefmark{1}State Key Laboratory of Networking and Switching Technology\\
\IEEEauthorrefmark{2}National Engineering Research Center for Mobile Network Technologies \\
\IEEEauthorrefmark{3}Beijing University of Posts and Telecommunications, Beijing, China \\
\IEEEauthorrefmark{4}Department of Broadband Communication, Peng Cheng Laboratory, Shenzhen, Guangdong, China\\
Email: \{zhangwenzhao, hanshujun, haixiao, smy\_bupt, wangbizhu\_7, xuxiaodong\}@bupt.edu.cn}
}

\maketitle

\begin{abstract}
Multi-vehicle collaborative perception (MvCP) is considered as a key technology to facilitate automated driving (AD), where real-time MvCP under limited resources is significant for reliable AD. In this paper, we formulate a lightweight acceleration scheme for intermediate-fusion (IF) MvCP, which can adapt to both situations of limited computation and communication resources. We provide a relaxed definition \textit{conditional additivity} and analyze the \textit{conditional additivity} for various DNN linear layers. On this basis, we focus on the IF-MvCP based on additive feature fusion, and derive the MvCP precision consistency of the forward and backward feature fusion position (FP) adjustments among linear layers. Through experiments, we further validate the precision consistency of the FP adjustment method. Moreover, we propose an \underline{F}P \underline{a}djustment among \underline{l}inear \underline{l}ayers (FALL) scheme for MvCP acceleration without precision loss theoretically. Simulation results show that the proposed FALL can reduce MvCP latency by up to 74.8\% under limited communication resources and by up to 30.3\% under limited computation resources.
\end{abstract}


\section{Introduction}
With the breakthroughs in artificial intelligence (AI), automated driving (AD) technology becomes a critical component of intelligent transportation systems (ITS), where accurate positioning and environment perception are significant for reliable intelligent driving \cite{gao2024survey}\cite{yang2024positioning}. Despite recent advances in single-vehicle perception as the development of multi-modal sensors and computer vision techniques, the challenge for accurate perception remains due to the occlusions and sparse sensor observations \cite{wang2025occlusion}. As a solution, researchers investigated multi-vehicle collaborative perception (MvCP) technology leveraging vehicle-to-vehicle (V2V) communication \cite{dao2024practical}.

This cooperative approach of MvCP can share sensing information among connected automated vehicles (CAVs), thereby compensating for the degradation in perception precision caused by their individual single-view limitations \cite{xu2022opv2v}. However, there are still challenges in deploying MvCP services in real-time due to the significant computational resource demands and dynamic wireless channel qualities of ITS \cite{li2023computation}. The traditional cloud-based processing paradigm can not satisfy the real-time requirement of MvCP, which may suffer congestion with massive data transmission \cite{zhang2025joint}. Besides, the edge devices (i.e., automated vehicles and road side units) are now equipped with more powerful computational capabilities, enabling them to process computation-intensive intelligent services \cite{liu2025truthful}. In this context, edge intelligence (EI) emerges as a promising solution to decrease transmission latency by processing intelligent services at the edge devices \cite{chen2023edge}.

To facilitate the deployment of intelligent perception services on the edge devices, several studies have been proposed to reduce the processing overhead of perception tasks, including communication overhead \cite{lu2025joint} \cite{wang2020v2vnet} and computation overhead \cite{lu2024crossprune} \cite{jafarpourmarzouni2024towards}. Specifically, the authors of \cite{lu2025joint} and \cite{wang2020v2vnet} introduced the sensor data compression method to reduce the communication overhead. Lu et al. \cite{lu2025joint} proposed a joint optimization problem of cooperative vehicles selection and compression ratio selection to reduce the size of sensing data, while guaranteeing the perception precision requirement. Wang et al. \cite{wang2020v2vnet} adapted the variational image compression algorithm to compress the intermediate representations, and then quantized and encoded the latent representation with few bits for transmission. Additionally, the authors of \cite{lu2024crossprune} and \cite{jafarpourmarzouni2024towards} designed pruning and quantization methods to reduce the computation overhead. Lu et al. \cite{lu2024crossprune} proposed a modal cooperative pruning framework designed for camera-LiDAR fused perception in autonomous driving, which attained superior pruning ratios while minimizing precision loss. S et al. \cite{jafarpourmarzouni2024towards} focused on optimizing model performance through integration of pruning and quantization techniques, achieving faster inference speed with minimal impact on precision. However, the aforementioned literature \cite{lu2025joint, wang2020v2vnet, lu2024crossprune, jafarpourmarzouni2024towards} achieved perception acceleration by sacrificing precision. Moreover, these approaches for perception acceleration considered the overhead decrease of computation or communication independently, and they are difficult to adapt simultaneously to both situations with limited computation resources and limited communication resources.

To tackle the above challenges, we propose a lightweight \underline{f}usion position (FP) \underline{a}djustment among \underline{l}inear \underline{l}ayers (FALL) scheme to accelerate MvCP in edge intelligence empowered ITS. The proposed FALL scheme can achieve MvCP acceleration under both limited computation and limited communication resources situations, which can obtain consistent perception precision with MvCP under the original FP theoretically. The main contributions are summarized as follows:
\begin{itemize}
\item We provide a relaxed definition \textit{conditional additivity} based on the concept of \textit{additivity}. Furthermore, we analyze the \textit{conditional additivity} of DNN linear layers and the DNN model consisting of multiple linear layers.
\item We derive the MvCP precision consistency of the forward and backward FP adjustments among linear layers. Additionally, we analyze the FP adjustment range of MvCP based on the PIXOR model, and validate the precision consistency under different FPs via experiments.
\item We propose the FALL scheme, which can achieve MvCP acceleration under both limited computation and limited communication resources situations without precision loss theoretically. Besides, we examine the acceleration performance of FALL under different transmission rates, showing a latency reduction of up to 74.8\%.
\end{itemize}

\section{System Model}
In this paper, we focus on the intermediate fusion (IF) based MvCP service\cite{xu2022opv2v}, which requires less transmission bandwidth than early fusion and provides more comprehensive information than late fusion \cite{tang2025rocooper}. The inference process of IF-MvCP can be concluded as three stages: 1) feature extraction (FE); 2) feature fusion (FF); 3) object detection (OD). As shown in Fig. \ref{system}, we consider a scenario of IF-MvCP based on additive feature fusion, containing one ego-vehicle and a set $\mathcal{N}=\{1,..., N\}$ of collaborative vehicles (co-vehicles). The vehicles can communicate with each other via PC5-based V2V sidelink\cite{38885}. Each vehicle is equipped with a computing unit for task processing, where the model before FP is deployed and processed at co-vehicles and the ego-vehicle in parallel (FE), and the model behind FP is deployed and processed at ego-vehicle centrally (OD) after additive FF.

\begin{figure}[b]
\centerline{\includegraphics[width=0.38\textwidth]{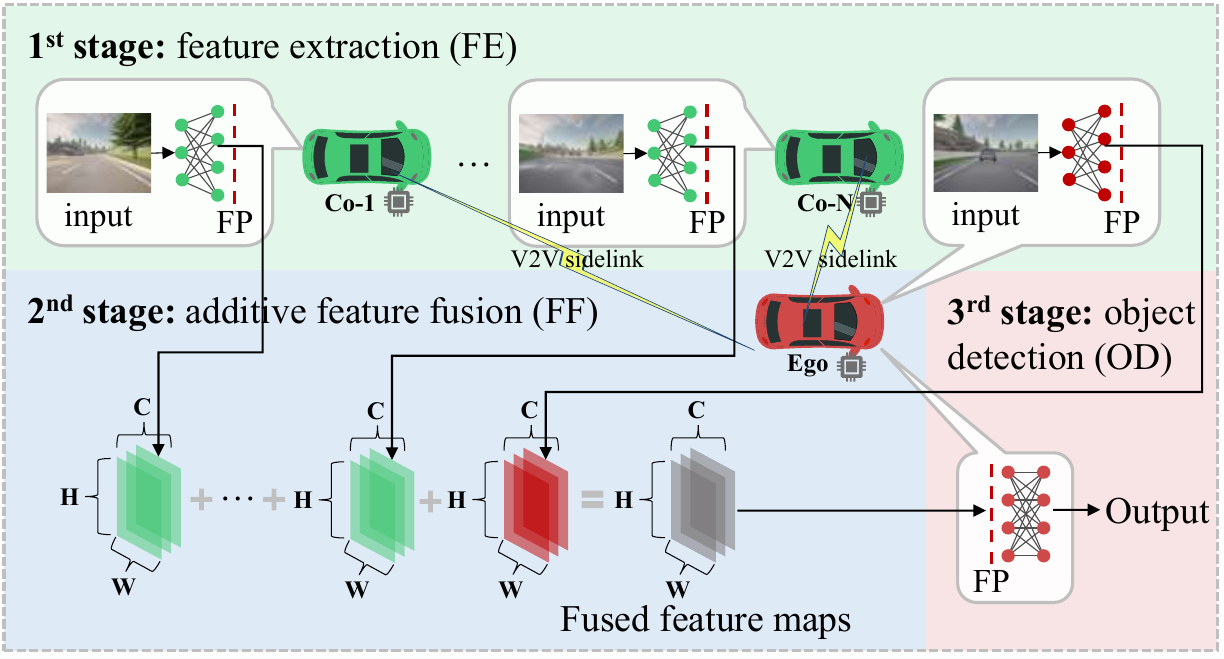}}
\caption{IF-MvCP with additive feature fusion system model.}
\label{system}
\end{figure}

Considering that there are more affine transformations rather than linear transformations within DNN inference, we define the layers based on the affine transformation (e.g., fully connected, convolution, and batch normalization) as linear layers. Note that linear transformations $\mathcal{F}$ satisfy additivity, i.e., $\mathcal{F}(\sum_{i=1}^{n}x_i) = \sum_{i=1}^{n}\mathcal{F}(x_i)$, while affine transmissions do not satisfy additivity because they consist of both linear transformations and translations. To extend the discussion of additivity to DNN inference, we provide a relaxed definition \textit{conditional additivity} drawn inspiration from the \textit{separable} concept introduced in \cite{sun2022serpens}, based on which we will construct the MvCP acceleration scheme. Specifically, the \textit{conditional additivity} is defined as follows.

\begin{definition}
\textbf{(Conditional Additivity)} The function $\mathcal{F}$ is conditional additive if exits functions $\mathcal{F}_i$, for any $x_i, i\in[1,n]$, satisfies $\mathcal{F}(\sum_{i=1}^{n}x_i)=\sum_{i=1}^{n}\mathcal{F}_i(x_i)$.
\end{definition}

\subsection{Conditional Additivity Analysis of DNN Linear Layers}
In this subsection, we provide the analysis of \textit{conditional additivity} for various DNN linear layers, including fully-connected, convolution, deconvolution, batch normalization, and average pooling.

\begin{proposition}
\textbf{(Conditional Additivity of Fully-Connected Layer)} The computation of fully-connected layer is an affine transformation, which is represented as ${{\mathsf{\mathcal{F}}}_{fc}}\left( \mathbf{x} \right)=\mathbf{x}\cdot \mathbf{W}+\mathbf{b}$. $\mathbf{W}$ is the weight matrix and $\mathbf{b}$ is the bias vector. Given $\mathbf{x}=\sum_{i=1}^{n}\mathbf{x}_i$, $\mathcal{F}_{fc}$ can be transformed as: $\mathcal{F}_{fc}(\mathbf{x}) 
= \mathcal{F}_{fc}\left( \sum_{i=1}^{n} \mathbf{x}_i\right) = \sum_{i=1}^{n} \mathbf{x}_i \cdot \mathbf{W} + \mathbf{b} = \sum_{i=1}^{n} \mathbf{x}_i \cdot \mathbf{W} + \sum_{i=1}^{n}\mathbf{b}_i = \sum_{i=1}^{n} \left( \mathbf{x}_i \cdot \mathbf{W} + \mathbf{b}_i \right) = \sum_{i=1}^{n} \mathcal{F}_{fc,i} \left( \mathbf{x}_i\right)$, where $\mathcal{F}_{fc,i} \left( \mathbf{x}_i\right)=\left( \mathbf{x}_i \cdot \mathbf{W} + \mathbf{b}_i \right)$, $\sum_{i=1}^n\mathbf{b}_i=\mathbf{b}$. Therefore, the fully-connected layer satisfies conditional additivity.
\end{proposition}

\begin{proposition}
\label{ca_conv}
\textbf{(Conditional Additivity of Convolution Layer)} The computation of each feature patch for a convolution layer is an affine transformation, which is represented as $\mathcal{F}_{conv}\left(\mathbf{x}\right)=\sum_c\sum_{w,h}\mathbf{x}_c\left(j+w,k+h\right)\cdotp\mathbf{K}_c\left(w,h\right)+b$. $\mathbf{x}_c$ is the input feature map of channel $c$, $\mathbf{K}_c$ is the kernel of channel $c$. $(j,k)$ represents the starting pixel of convolution, $(w,h)$ represents the element pixel of kernel, $b$ is the bias. Given $\mathbf{x}=\sum_{i=1}^{n}\mathbf{x}_i$, $\mathcal{F}_{conv}$ can be transformed as: $\mathcal{F}_{conv}(\mathbf{x}) 
= \mathcal{F}_{conv}(\sum_{i=1}^{n}\mathbf{x}_i) = \sum_{i=1}^{n} \left( \sum_{c} \sum_{w,h} \mathbf{x}_{c,i}(j + w, k + h) \cdot \mathbf{K}_c(w,h) + b_{i} \right) = \sum_{i=1}^{n} \mathcal{F}_{conv,i}(\mathbf{x}_i)$, where $\mathcal{F}_{conv,i}(\mathbf{x}_i)=\sum_{c} \sum_{w,h} \mathbf{x}_{c,i}(j + w, k + h) \cdot \mathbf{K}_c(w,h) + b_{i}$, $\sum_{i=1}^nb_i=b$. Therefore, the convolution layer satisfies conditional additivity. Notably, if the hyperparameter padding $p$ of $\mathcal{F}$ is not zero, the padding $p_i$ of $\mathcal{F}_{conv,i}$ should satisfy $\sum_{i=1}^np_i=p$. 
\end{proposition}

\begin{proposition}
\label{ca_deconv}
\textbf{(Conditional Additivity of Deconvolution Layer)} The computation of deconvolution is the same as convolution, which can be represented as $\mathcal{F}_{deconv}\left(\mathbf{x}\right)=\mathcal{F}_{conv}\left(\mathbf{x}\right)$. The main difference between deconvolution and convolution is the size of output, which has no affect on the computation process. Therefore, the deconvolution layer satisfies conditional additivity as well, where $\mathcal{F}_{deconv,i}(\mathbf{x}_i)=\mathcal{F}_{conv,i}(\mathbf{x}_i)$.
\end{proposition}

\begin{proposition}
\textbf{(Conditional Additivity of Batch Normalization Layer)} The computation of batch normalization during inference simplifies to an affine transformation, which is represented as ${{\mathsf{\mathcal{F}}}_{bn}}\left( \mathbf{x}\right)=\gamma \left( \frac{\mathbf{x}-\mu }{\sqrt{{{\delta }^{2}}+\varepsilon}} \right)+\beta$. $\mu$ and $\delta^{2}$ are the mean and variance counted according to the training data, $\gamma$ and $\beta$ are trainable parameters, and $\varepsilon$ is a small constant avoiding division by zero. The above parameters are fixed during inference. Given $\mathbf{x}=\sum_{i=1}^{n}\mathbf{x}_i$, $\mathcal{F}_{bn}$ can be transformed as: $\mathcal{F}_{bn}\left(\mathbf{x}\right) 
=\mathcal{F}_{bn}\left(\sum_{i=1}^{n}\mathbf{x}_i\right) = \gamma \left( \frac{\sum_{i=1}^{n} \mathbf{x}_i - \sum_{i=1}^{n} \mu_i}{\sqrt{\delta^2 + \varepsilon}} \right) + \sum_{i=1}^{n} \beta_i = \sum_{i=1}^{n} \left( \frac{\gamma (\mathbf{x}_i - \mu_i)}{\sqrt{\delta^2 + \varepsilon}} + \beta_i \right) = \sum_{i=1}^{n} \mathcal{F}_{bn,i}\left( \mathbf{x}_i\right)$, where $\mathcal{F}_{bn,i}\left( \mathbf{x}_i\right)=\left( \frac{\gamma (\mathbf{x}_i - \mu_i)}{\sqrt{\delta^2 + \varepsilon}} + \beta_i \right)$, $\sum_{i=1}^{n}\beta_i=\beta$, $\sum_{i=1}^{n}\mu_i=\mu$. Therefore, the batch normalization layer satisfies conditional additivity.
\end{proposition}

\begin{proposition}
\textbf{(Conditional Additivity of Average Pooling)} The computation of average pooling is a linear transformation. To simplify the illustration, we only analyze the calculation within a single pooling window, which is represented as $\mathcal{F}_{ap}\left(\mathbf{x}\right)=\frac{1}{WH}\sum_{w=1:W,h=1:H}\mathbf{x}\left(j+w,k+h\right)$. $(j,k)$ is the starting pixel coordinate of the average pooling, $W \times H$ is the size of pooling window. Given $\mathbf{x}=\sum_{i=1}^n\mathbf{x}_i$, $\mathcal{F}_{ap}$ can be transformed as: $\mathcal{F}_{ap}(\mathbf{x}) 
= \mathcal{F}_{ap}(\sum_{i=1}^n\mathbf{x}_i) = \frac{1}{WH} \sum_{w=1:W,\, h=1:H} \sum_{i=1}^{n} \mathbf{x}_i(j+w, k+h) = \sum_{i=1}^{n} \left( \frac{1}{WH} \sum_{w=1:W,\, h=1:H} \mathbf{x}_i(j+w, k+h) \right) = \sum_{i=1}^{n} \mathcal{F}_{ap,i}(\mathbf{x}_i)$, where $\mathcal{F}_{ap,i}(\mathbf{x}_i) = \mathcal{F}_{ap}(\mathbf{x}_i)$. Therefore, the average pooling layer satisfies conditional additivity (more precisely, it satisfies additivity).
\end{proposition}

\subsection{Conditional Additivity Analysis of DNN Model Consisting of Multiple Linear Layers}
On the one hand, the DNN model has a multi-layer structure, where the output of the former layer is the input of the latter one. Thus, the DNN inference can be considered as a composition function. On the other hand, there are many shortcut and skip connection structures in the DNN model \cite{he2016deep}, which can be regarded as a linear combination function. In this section, we will present the \textit{conditional additivity} analysis for the composition and linear combination of conditional additive functions (CAFs).
\begin{thm}
\label{com}
\textbf{(Composition of CAFs)} If functions $\mathcal{F}_1$, $\mathcal{F}_2$ satisfy conditional additivity, their composition $\mathcal{F}_1(\mathcal{F}_2(\cdot))$ satisfies conditional additivity:
\begin{equation}
    \mathcal{F}_1(\mathcal{F}_2(\sum_{i=1}^nx_i))=\sum_{i=1}^{n}{\mathcal{F}_{1,i}(\mathcal{F}_{2,i}(x_i))}.
\end{equation}
\end{thm}
\begin{proof}
$\mathcal{F}_1(\mathcal{F}_2(\sum_{i=1}^nx_i)) = \mathcal{F}_1(\sum_{i=1}^n\mathcal{F}_{2,i}(x_i))$ \\ \hspace*{3.6cm}  $= \sum_{i=1}^{n}{\mathcal{F}_{1,i}(\mathcal{F}_{2,i}(x_i))}.$
\end{proof}
\begin{thm}
\label{lc}
\textbf{(Linear combination of CAFs)} If functions $\mathcal{F}_1$, $\mathcal{F}_2$ satisfy conditional additivity, their linear combination $\alpha\mathcal{F}_1(\cdot)+\beta\mathcal{F}_2(\cdot)$ satisfies conditional additivity:
\begin{equation}
\alpha\mathcal{F}_1(\sum_{i=1}^nx_i)+\beta\mathcal{F}_2(\sum_{i=1}^nx_i) = \sum_{i=1}^n(\alpha\mathcal{F}_{1,i}(x_i) + \beta\mathcal{F}_{2,i}(x_i)).
\end{equation}
\end{thm}
 \begin{proof}
$\alpha\mathcal{F}_1(\sum_{i=1}^nx_i)+\beta\mathcal{F}_2(\sum_{i=1}^nx_i) = \sum_{i=1}^n(\alpha\mathcal{F}_{1,i}) + \sum_{i=1}^n(\beta\mathcal{F}_{2,i}) = \sum_{i=1}^n(\alpha\mathcal{F}_{1,i}(x_i) + \beta\mathcal{F}_{2,i}(x_i))$.
\end{proof}

\section{Precision Consistency of Fusion Position Adjustment Among Linear Layers}
Based on the analysis for \textit{conditional additivity} of DNN linear layers, it can be derived that the FP of IF-MvCP can be adjusted among the linear layers with the same inference precision. In this section, we first derive the precision consistency of the forward and backward FP adjustments, respectively. Afterward, we validate the precision consistency of MvCP service based on the PIXOR model under different FP adjustments.

\subsection{Precision Consistency of FP Adjusted Forward}
Given a trained MvCP model, we denote the original FP as $L_{original}$. Besides, we denote the adjusted forward FP as $L_{forward}$, which is ahead of $L_{original}$. Fig. \ref{fp_forward} illustrates the inference processes under FP at $L_{original}$ and $L_{forward}$. If the intermediate features are fused at $L_{original}$, the output $o$ at $L_{original}$ can be calculated as
\begin{equation}
o=\sum_{i=0}^N{\mathcal{F}_i{(f_i)}},
\end{equation}
where $f_0$ represents the intermediate feature at $L_{forward}$ of ego-vehicle, and $f_i (1 \le i\le N)$ represents that of co-vehicle $i$. $\mathcal{F}_0$ represents the model between $L_{forward}$ and $L_{original}$ deployed at the ego-vehicle, and $\mathcal{F}_i (1 \le i\le N)$ represents that deployed at co-vehicle $i$.

\begin{figure}[b]
\centerline{\includegraphics[width=0.4\textwidth]{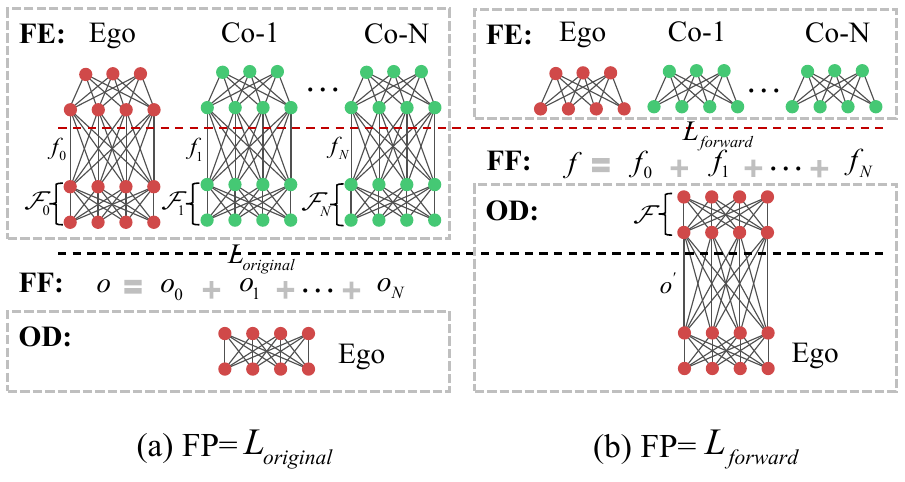}}
\caption{The inference processes under FP at $L_{original}$ and $L_{forward}$.}
\label{fp_forward}
\end{figure}

If the perception features are fused at $L_{forward}$, the output $o^{'}$ at $L_{original}$ can be calculated as
\begin{equation}
o^{'}=\mathcal{F}(\sum_{i=0}^N(f_i)),
\end{equation}
where $\mathcal{F}$ represents the model between $L_{forward}$ and $L_{original}$ deployed at the ego-vehicle centrally.

According to the definition of \textit{conditional additivity}, if $\mathcal{F}$ is conditional additive, the outputs at $L_{original }$ under both cases of FP at $L_{original}$ and $L_{forward}$ are consistent (i.e., $o=o^{'}$). Therefore, we can derive the precision consistency for the case of forward FP adjustment.
\begin{thm}
\label{consis_forward}
For an MvCP service with the original FP $L_{original}$, and $L_{forward}$ ahead of $L_{original}$, assume that the model $\mathcal{F}$ between $L_{forward}$ and $L_{original}$ satisfies conditional additivity. Then, the perception precision under $FP=L_{orignal}$ and $FP=L_{forward}$ is consistent.
\end{thm}
\begin{proof}
The inference results under $FP=L_{orignal}$ and $FP=L_{forward}$ are denoted as $r$ and $r^{'}$, respectively. The model after $L_{original}$ is denoted as $\mathcal{F}_1$. Since $o^{'}=\mathcal{F}(\sum_{i=0}^N(f_i))=\sum_{i=0}^N\mathcal{F}_i(f_i)=o$, $r=\mathcal{F}_1(o)=\mathcal{F}_1(o^{'})=r^{'}$. Therefore, Theorem \ref{consis_forward} is proved.
\end{proof}

\subsection{Precision Consistency of FP Adjusted Backward}
We denote the adjusted backward FP as $L_{backward}$, which is behind $L_{original}$. Fig. \ref{fp_backward} illustrates the inference processes under feature fusion at $L_{original}$ and $L_{backward}$. If the perception features are fused at $L_{original}$, the output $b$ at $L_{backward}$ can be calculated as
\begin{equation}
b = \mathcal{F}^{'}(\sum_{i=0}^No_i),
\end{equation}
where $o_0$ represents the output at $L_{original}$ of ego-vehicle, and $o_i (1\le i\le N)$ represents that of co-vehicle $i$. $\mathcal{F}^{'}$ represents the model between $L_{original}$ and $L_{backward}$ deployed at the ego-vehicle centrally.

\begin{figure}[b]
\centerline{\includegraphics[width=0.4\textwidth]{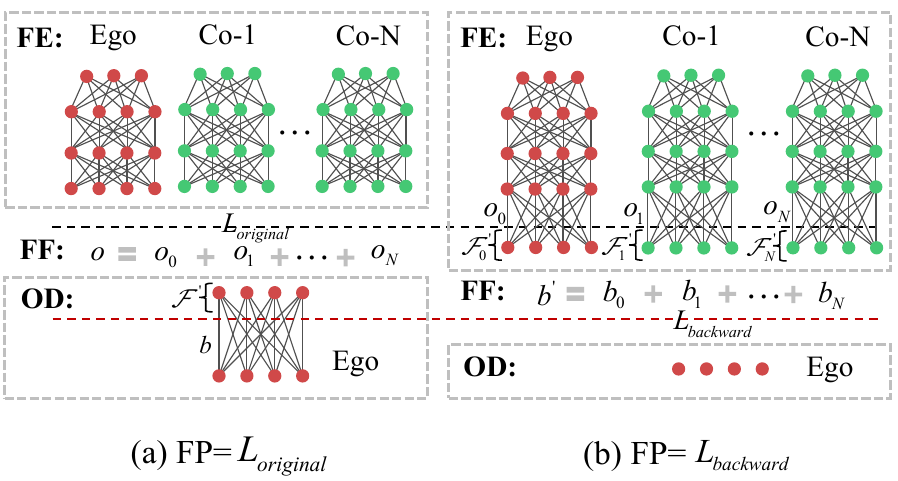}}
\caption{The inference processes under FP at $L_{original}$ and $L_{backward}$.}
\label{fp_backward}
\end{figure}

If the perception features are fused at $L_{backward}$, the output $b^{'}$ at $L_{backward}$ can be calculated as
\begin{equation}
b^{'}=\sum_{i=0}^{N}\mathcal{F}^{'}_{i}(o_i),
\end{equation}
where $\mathcal{F}^{'}_{0}$ represents the model between $L_{original}$ and $L_{backward}$ deployed at the ego-vehicle, and $\mathcal{F}^{'}_{i} (1\le i\le N)$ represents that deployed at co-vehicle $i$.

According to the definition of \textit{conditional additivity}, if $\mathcal{F}^{'}$ is conditional additive, the outputs at $L_{backward}$ under both cases of FP at $L_{original}$ and $L_{backward}$ are consistent (i.e., $b=b^{'}$). Therefore, we can derive the precision consistency for the case of backward FP adjustment.
\begin{thm}
\label{consis_backward}
For an MvCP service with the original FP $L_{original}$, and $L_{backward}$ behind $L_{original}$, assume that the model $\mathcal{F}^{'}$ between $L_{original}$ and $L_{backward}$ satisfies conditional additivity. Then, the perception precision under $FP=L_{original}$ and $FP=L_{backward}$ is consistent.
\end{thm}
\begin{proof}
The inference results under $FP=L_{original}$ and $FP=L_{backward}$ are denoted as $r$ and $r^{'}$, respectively. The model after $L_{backward}$ is denoted as $\mathcal{F}_2$. Since $b=\mathcal{F}^{'}(\sum_{i=0}^No_i)=\sum_{i=0}^N\mathcal{F}_i^{'}(o_i)=b^{'}$, $r=\mathcal{F}_2(b)=\mathcal{F}_2(b^{'})=r^{'}$. Therefore, Theorem \ref{consis_backward} is proved.
\end{proof}

\subsection{Precision Consistency Validation for FP Adjustment}
\begin{figure}[b]
\centerline{\includegraphics[width=0.32\textwidth]{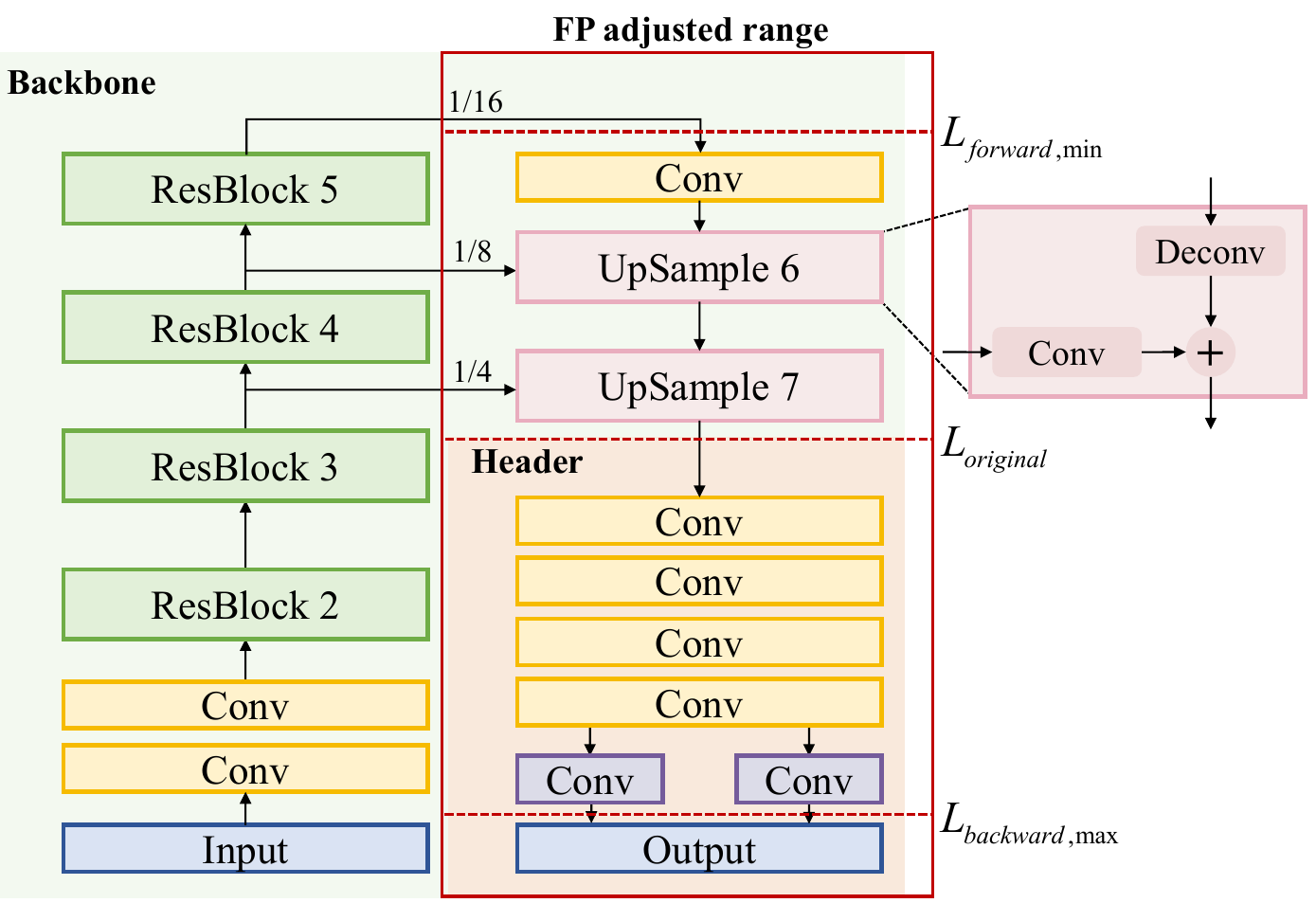}}
\caption{The architecture of PIXOR.}
\label{pixor}
\end{figure}

Specifically, we utilize the MvCP service based on the state-of-the-art model PIXOR \cite{yang2018pixor} to validate the precision consistency of FP adjustment. Fig. \ref{pixor} illustrates the architecture of PIXOR, which consists of a backbone model for feature extraction and a header model for object detection.

It is shown that the UpSample layer of PIXOR is composed of the addition of a convolution and a deconvolution, which is represented as $\mathcal{F}_{upsample}(\mathbf{x})=\mathcal{F}_{conv}(\mathbf{x})+\mathcal{F}_{deconv}(\mathbf{x})$. Because $\mathcal{F}_{conv}$ and $\mathcal{F}_{deconv}$ satisfy \textit{conditional additivity} (refer to Proposition \ref{ca_conv} and \ref{ca_deconv}), it can be derived that $\mathcal{F}_{upsample}$ satisfy \textit{conditional additivity} referring to Theorem \ref{lc}. In addition, since the ResBlocks contain nonlinear structure (i.e., ReLU), they do not satisfy \textit{conditional additivity}. According to Theorem \ref{com}, the composition of CAFs satisfies \textit{conditional additivity}. Therefore, the FP of PIXOR can be adjusted among the linear layers from $L_{forward, min}$ to $L_{backward, max}$ (as shown in Fig. \ref{pixor}) with consistent precision, which can be drawn from Theorem \ref{consis_forward} and \ref{consis_backward}.

Fig. \ref{p_a} shows the detection results of MvCP with three CAVs (including one ego-vehicle and two co-vehicles) under the adjusted forward FP, the original FP, and the adjusted backward FP. The average precisions (AP) at different intersection-over-union (IoU) threshold under different adjusted FPs within the linear layers are illustrated in Table \ref{t_a}. We denote FP adjusted forward $i$ layers as $f_{i}$, and FP adjusted backward $i$ layers as $b_{i}$. The results indicate that the MvCP precision under different adjusted FPs within linear layers is approximately consistent with that under the original FP (with the maximum error not exceeding 0.05), which can further validate the theoretical derivation for the precision consistency of FP adjustment. 


\begin{figure}[tb]
	\centering
	\begin{minipage}[c]{0.17\textwidth}
		\centering
		\includegraphics[width=0.58\textwidth]{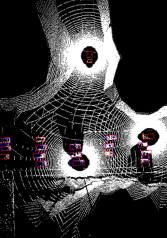}
		\subcaption{\footnotesize Adjusted forward FP}
	\end{minipage} 
	\begin{minipage}[c]{0.1\textwidth}
		\centering
		\includegraphics[width=\textwidth]{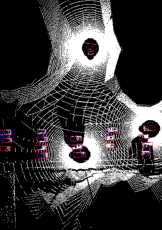}
		\subcaption{\footnotesize Original FP}
	\end{minipage} 
	\begin{minipage}[c]{0.17\textwidth}
		\centering
		\includegraphics[width=0.57\textwidth]{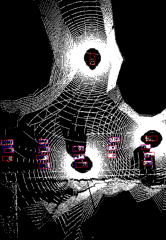}
		\subcaption{\footnotesize Adjusted backward FP}
	\end{minipage}
	\caption{The detection results of IF-MvCP with three CAVs under different FP.}
	\label{p_a}
\end{figure}

\begin{table}[b]
\scriptsize
\centering
\caption{Average precisions of different fusion positions.}
\begin{tabular}{|c|c|c|c|}
  \hline
  \textbf{IoU} & \textbf{0.3} & \textbf{0.5} & \textbf{0.7}\\
  \hline
  original & \underline{0.88} & \underline{0.85} & \underline{0.64} \\
  \hline
  $f_3$ & 0.87 & 0.84 & \textbf{0.63} \\
  $f_2$ & 0.84 & 0.82 & 0.60 \\
  $f_1$ & 0.84 & 0.82 & 0.60 \\
  \hline
  $b_1$ & 0.84 & 0.81 & 0.61 \\
  $b_2$ & 0.83 & 0.81 & 0.60 \\
  $b_3$ & 0.83 & 0.81 & 0.60 \\
  $b_4$ & \textbf{0.88} & \textbf{0.86} & 0.59 \\
  $b_5$ & \textbf{0.88} & \textbf{0.86} & 0.59 \\
  \hline
\end{tabular}
\label{t_a}
\end{table}

\section{MvCP Acceleration Scheme Based on Fusion Position Adjustment Without Precision Loss}

Based on the above analysis for the precision consistency of FP adjustment, we propose the lightweight MvCP acceleration scheme based on \underline{F}P \underline{a}djustment among \underline{l}inear \underline{l}ayers (FALL) without precision loss. Specifically, the FP can be dynamically adjusted according to the system resource situation to achieve MvCP acceleration. Subsequently, we analyze the acceleration capability of the FALL scheme using MvCP based on PIXOR as an example. Fig. \ref{pixor_para} shows the computation workload and intermediate feature size of each layer in PIXOR.
\begin{figure}[b]
\centerline{\includegraphics[width=0.32\textwidth]{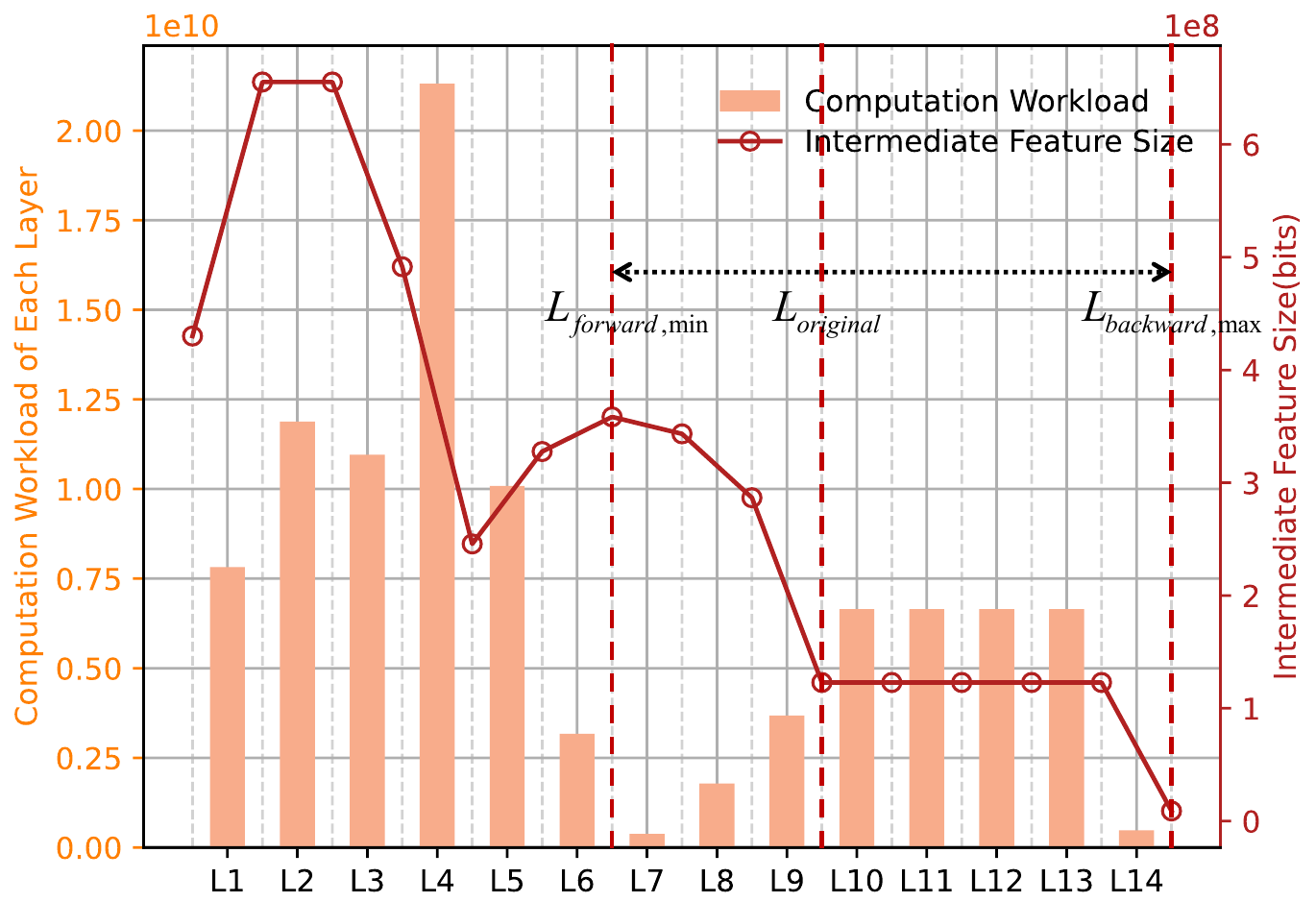}}
\caption{The computation workload and intermediate feature size of each layer in PIXOR.}
\label{pixor_para}
\end{figure}
\subsection{Computation Latency}
Considering the randomness of co-vehicle selection, the computation resources of co-vehicles are variable and potentially lower than ego-vehicle. In general, we assume that the computation resource of each co-vehicle $f_{co,i} (1\le i \le N$) is less than that of the ego-vehicle $f_{ego}$ (i.e., $f_{co,i} < f_{ego}$). Thus, the former FP corresponds to less computation latency, which can be derived as follows. Firstly, we denote the computation latency of FP as $T_{c}(FP)$, which is calculated by
\begin{align}
T_{c}(FP) &= \max(\frac{\sum_{j=1}^{FP}C_j}{f_{co,i}}, \frac{\sum_{j=1}^{FP}C_j}{f_{ego}}) + \frac{\sum_{j=FP+1}^{L_{max}}C_j}{f_{ego}} \nonumber\\
&= \frac{\sum_{j=1}^{FP}C_j}{f_{co,min}} + \frac{\sum_{j=FP+1}^{L_{max}}C_j}{f_{ego}},
\end{align}
where $C_j$ represents the computation workload of the $j$-th layer in PIXOR, $L_{max}$ represents the last layer of PIXOR. We consider two different FPs $L_1$ and $L_2$, which satisfy $L_1 < L_2$. Then, the comparison of computation latency under FP=$L_1$ and FP=$L_2$ is discussed as follows:

\begin{align}
T_{c}(L_1) - T_{c}(L_2) &= \frac{\sum_{j=1}^{L_1}C_j}{f_{co,min}} + \frac{\sum_{j=L_1+1}^{L_{max}}C_j}{f_{ego}} \nonumber\\
&- \left(\frac{\sum_{j=1}^{L_2}C_j}{f_{co,min}} + \frac{\sum_{j=L_2+1}^{L_{max}}C_j}{f_{ego}} \right) \nonumber\\
& = -\frac{\sum_{j=L_1+1}^{L_2}C_j}{f_{co,min}} + \frac{\sum_{j=L_1+1}^{L_2}C_j}{f_{ego}} \nonumber\\
& = \sum_{j=L_1+1}^{L_2}C_j\left(\frac{f_{co,min}-f_{ego}}{f_{co,min}\cdot f_{ego}} \right) < 0.
\end{align}

Therefore, $T_c(L_1) < T_c(L_2)$, which means \textit{the computation latency of a former FP is less than that of a latter FP}.
\subsection{Transmission Latency}
The transmission latency of FP $T_t(FP)$ is calculated by
\begin{equation}
    T_t(FP) = \frac{I_{FP}}{R},
\end{equation}
where $I_{FP}$ represents the intermediate feature size at FP, and $R$ represents the transmission rate of the V2V sidelink. It can be seen from Fig. \ref{pixor_para} that the intermediate feature size for PIXOR of the latter layer is mostly less than (or equal to, such as $L_{10} \sim L_{13}$) that of the former layer. This feature is also generally applicable in other models \cite{liu2024dnn}. Thus, we represent the intermediate feature size of the former FP $L_1$ as $I_1$ and that of the latter FP $L_2$ as $I_2$, which satisfy $I_1 \ge I_2$. Thus, $T_t(L_1) = \frac{I_1}{R} \ge T_t(L_2) = \frac{I_2}{R}$, which means \textit{the transmission latency of a former FP is larger than or equal to that of a latter FP}.

\subsection{Total Latency}
The total latency of MvCP $T_{total}$ is composed of the computation and transmission latency, which is calculated as $T_{total} = T_c + T_t$. From the above discussion about the effect of FP adjustment on $T_c$ and $T_t$, it can be observed that the proposed FALL can reduce $T_c$ and $T_t$, while there is a trade-off between them. Therefore, the optimal FP with minimum latency differs depending on whether the limitation is on computation resources or communication resources. For example, if the computation resources become the performance bottleneck, the optimal FP tends to favor the former layer, while if the communication resources become the performance bottleneck, the optimal FP tends to favor the latter layer. The specific evaluation of FALL for MvCP acceleration under different resource limitations is provided in Section \ref{sim}.

\section{Performance Evaluation}
\label{sim}
In this section, we present the simulation results to compare the acceleration performance of our proposed FALL under different transmission rates. Specifically, the simulation is carried out based on PIXOR, where each parameter size is set as 4 Bytes. The computation resource of each co-vehicle is set to 0.5 TOPS, and the computation resource of the ego-vehicle is set to 30 TOPS.

\begin{figure*}[htb]
	\centering
	\begin{minipage}[c]{0.3\textwidth}
		\centering
		\includegraphics[width=\textwidth]{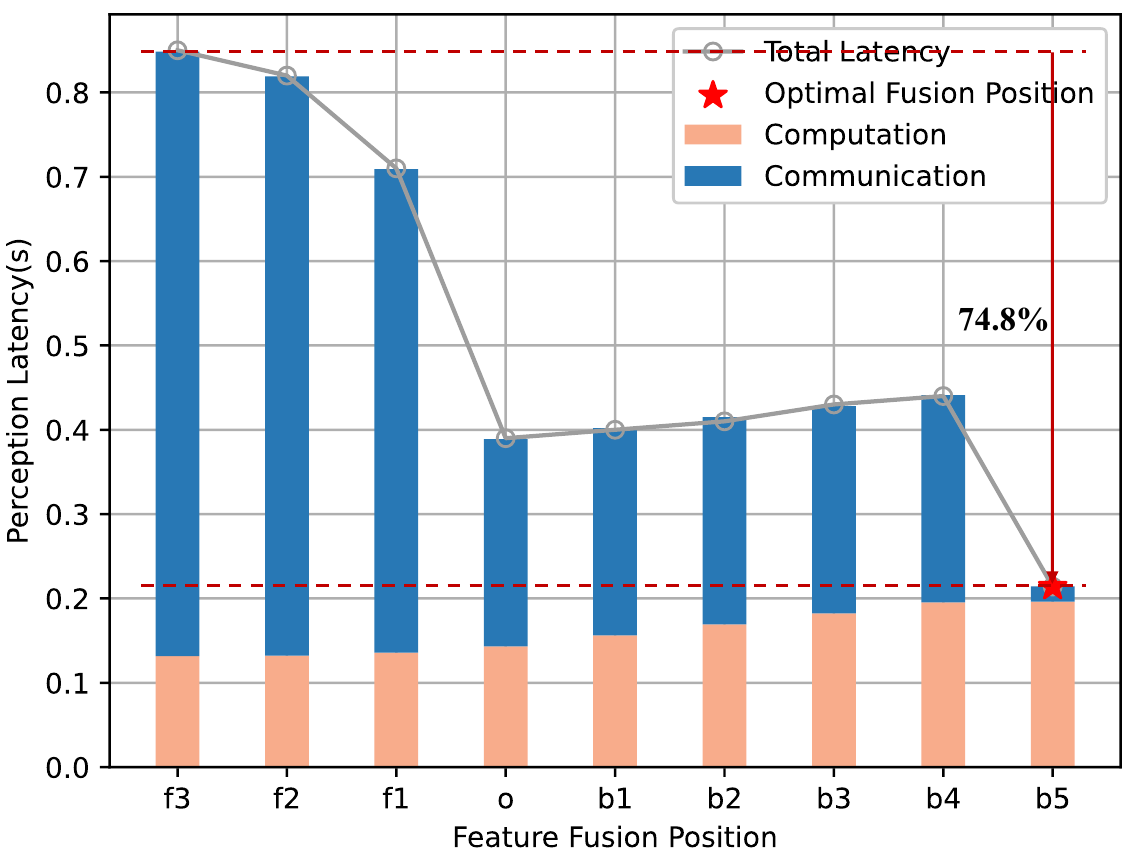}
		\subcaption{\footnotesize $R = 5\times 10^{8}$(bps)}
	\end{minipage} 
	\begin{minipage}[c]{0.3\textwidth}
		\centering
		\includegraphics[width=\textwidth]{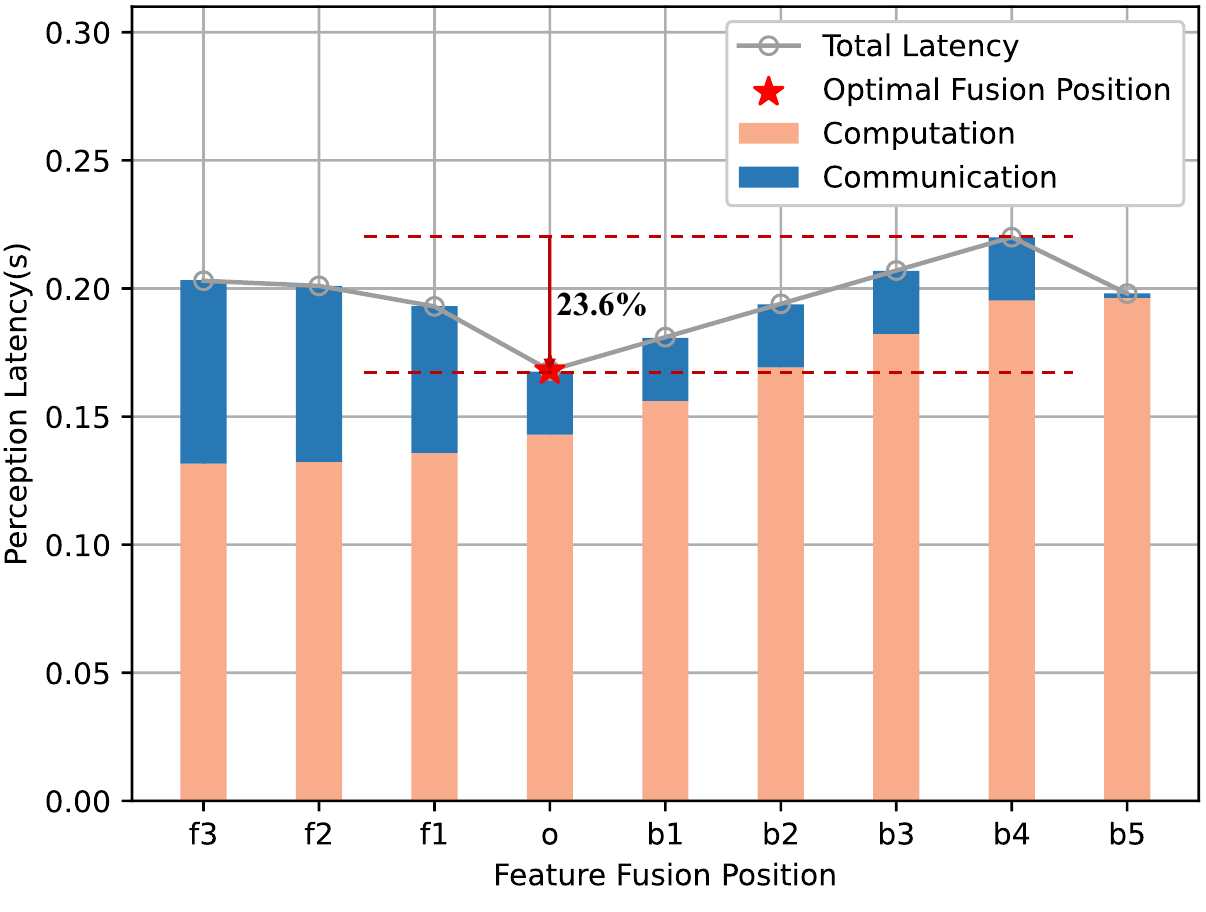}
		\subcaption{\footnotesize $R = 5\times 10^{9}$(bps)}
        \label{p_latency_b}
	\end{minipage} 
	\begin{minipage}[c]{0.3\textwidth}
		\centering
		\includegraphics[width=\textwidth]{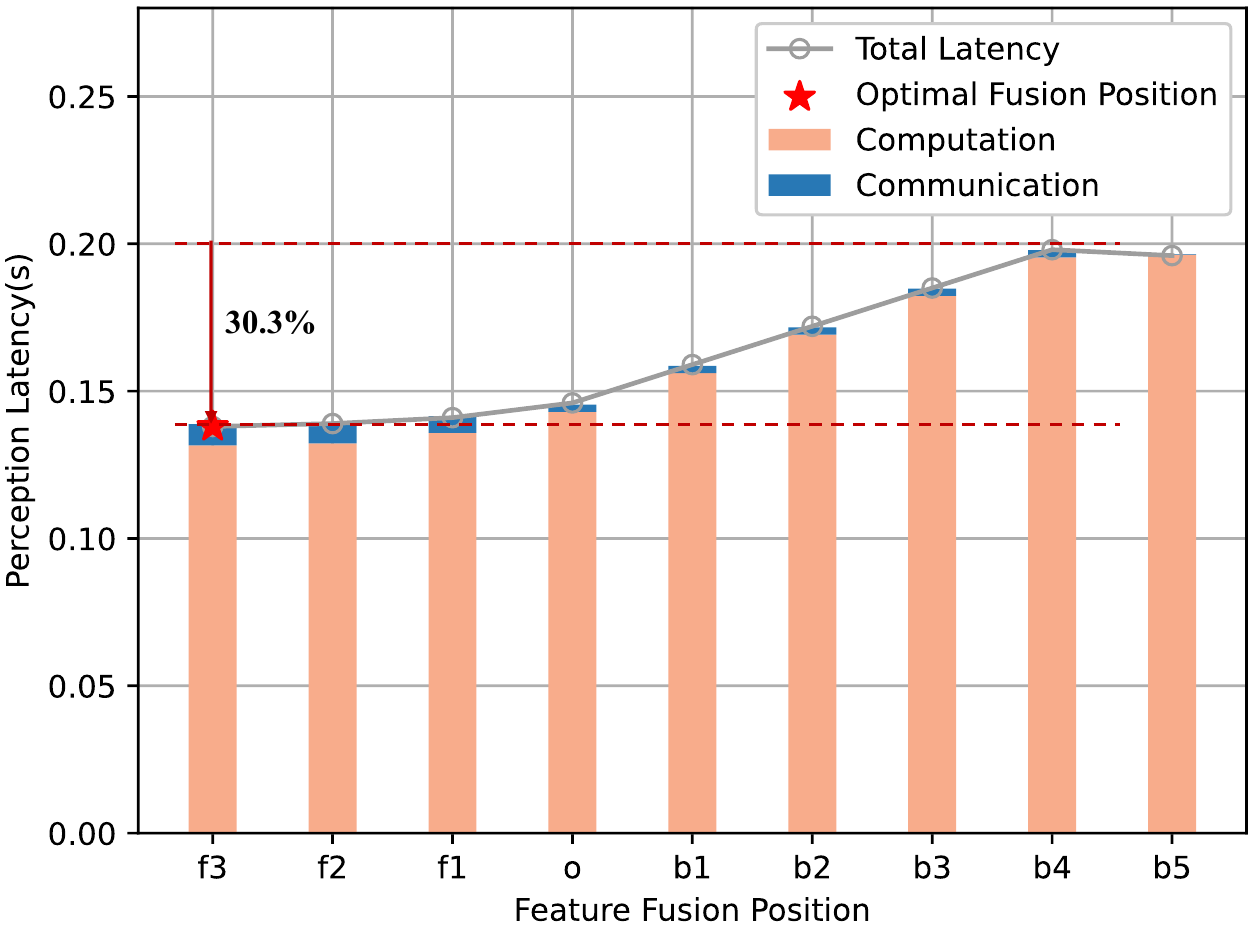}
		\subcaption{\footnotesize $R = 5\times 10^{10}$(bps)}
	\end{minipage}
	\caption{The total perception latency of different FPs under transmission rates $R$ ranging from $5\times10^{8}$ bps to $5\times10^{10}$ bps.}
	\label{p_latency}
\end{figure*}

Fig. \ref{p_latency} shows the total perception latency of different FP under transmission rates $R$ ranging from $5\times 10^{8}$ bps to $5\times 10^{10}$ bps. When $R=5\times 10^{8}$ bps, the communication latency becomes the performance bottleneck. In this case, the optimal FP is $b_5$ with the minimum transferred feature size, which can reduce the total latency by 74.8\% compared to the maximum value at $f_3$. When $R=5\times 10^{10}$ bps, the computation latency becomes the performance bottleneck. In this case, the optimal FP is $f_3$, where the most linear layers are processed at the ego-vehicle with more computation resources than co-vehicles. The total latency can be reduced by 30.3\% compared to the maximum value at $b_4$. When $R=5\times 10^{9}$ bps, the communication latency and computation latency are relatively close. Thus, the optimal FP depends on the trade-off between the communication latency and computation latency, as shown in Fig. \ref{p_latency_b} at the original FP $o$. In this case, the total latency can be reduced by 23.6\% compared to the maximum value at $b_4$.

\section{Conclusion}
In this paper, we investigated a lightweight acceleration scheme for IF-MvCP based on additive feature fusion. Firstly, the analysis of the \textit{conditional additivity} for various DNN linear layers and the DNN model consisting of multiple linear layers was presented. Besides, the precision consistency of the FP adjustment among linear layers was derived. Furthermore, the FALL scheme was proposed to accelerate MvCP while maintaining the perception precision, which can adapt to both situations of limited computation and communication resources. Simulation results validated the effectiveness of the proposed FALL under different limited resource situations.



\bibliographystyle{IEEEtran}
\bibliography{bib}

\end{document}